\documentclass[letterpaper, 10 pt, conference]{ieeeconf}  %
\IEEEoverridecommandlockouts                              %
\usepackage[T1]{fontenc} 
\usepackage{enumerate}
\usepackage{amsfonts}
\usepackage[colorinlistoftodos]{todonotes}
\usepackage[font=small,skip=12pt]{caption}
\usepackage{float}
\usepackage{subcaption}
\usepackage{paralist}
\usepackage{multirow}
\makeatletter
\let\NAT@parse\undefined
\makeatother
\usepackage{hyperref}
\usepackage{verbatim}
\usepackage{amsmath}
\usepackage{algorithm} 
\usepackage{algpseudocode}
\newtheorem{lemma}{Lemma}
\newtheorem{theorem}{Theorem}
\newtheorem{definition}{Definition}

\newcommand{\real}{\ensuremath{\mathbb{R}}}

\newcommand{\realnonnegative}{\ensuremath{\mathbb{R}}_{\ge 0}}

\newcommand{\union}{\operatorname{\cup}}

\newcommand{\subscr}[2]{#1_{\textup{#2}}}

\renewcommand{\epsilon}{\varepsilon}

\newcommand{\argmin}{\ensuremath{\operatorname{argmin}}}

\usepackage{psfrag}

\newcommand{\longthmtitle}[1]{\mbox{}{\bf \textit{(#1).}}}

\title{A Numerical Verification Framework for Differential Privacy in
  Estimation}
\author{Yunhai Han$^{1,2}$ and Sonia
  Mart{\'i}nez$^{1}$
  \thanks{$^{1}$Department of Mechanical and Aerospace Engineering,
    University of California at San Diego, La Jolla, CA, 92093, USA
    (email: y8han@ucsd.edu; soniamd@ucsd.edu).}
  \thanks{$^{2}$Laboratory for Intelligent Decision and Autonomous
    Robots, George W. Woodruff School of Mechanical Engineering,
    Georgia Institute of Technology, Atlanta, GA, 30313, USA (email:
    yhan364@gatech.edu). }  \thanks{This work was supported by Grant
    ONR N00014-19-1-2471.}  }%

\begin{document}
\maketitle
\thispagestyle{empty}
\pagestyle{empty}
\begin{abstract}
This work proposes an algorithmic method to verify differential privacy for estimation mechanisms with performance guarantees.  Differential privacy makes it hard to distinguish outputs of a mechanism produced by adjacent inputs. While obtaining theoretical conditions that guarantee differential privacy may be possible, evaluating these conditions in practice can be hard. This is especially true for estimation mechanisms that take values in continuous spaces, as this requires checking for an infinite set of inequalities. Instead, our verification approach consists of testing the differential privacy condition for a suitably chosen finite collection of events at the expense of some information loss.  More precisely, our data-driven, test framework for continuous range mechanisms first finds a highly-likely, compact event set, as well as a partition of this event, and then evaluates differential privacy wrt this partition. This results into a type of differential privacy with high confidence, which we are able to quantify precisely. This approach is then used to evaluate the differential-privacy properties of the recently proposed $W_2$ Moving Horizon Estimator. We confirm its properties, while comparing its performance with alternative approaches in simulation.
\end{abstract}
\section{INTRODUCTION} \label{section: introduction} 
A growing number of emerging, on-demand applications, require data
from users or sensors in order to make predictions and/or
recommendations. Examples include smart grids, traffic networks, or
home assistive technology. 
While more accurate information can {\color{black}benefit the quality of service}, an important concern is that sharing personalized
data may compromise the privacy of its users.  
This has been demonstrated over Netflix datasets~\cite{AN-VS:06}, as
well as on traffic monitoring
systems~\cite{BH-TI-QJ-DW-AMB-RH-JH-MG-MA-JB:12}.

Initially proposed from the database literature, \textit{Differential
  Privacy} \cite{CD-MF-NK-SA:06} addresses this issue, and has become
a standard in privacy specification 
of commercial products. 
More recently, differential privacy has attracted the attention of the
Systems and Control literature~\cite{JC-GED-SH-JLN-SM-GJP:16} and {\color{black}been} applied on control systems \cite{YW-ZH-SM-GED:14}, optimization
\cite{EN-PT-JC-2019}, and estimation and filtering \cite{JLN:14}.  In
particular, the work \cite{JLN-GJP:14} develops the concept of
differential privacy for Kalman filter design. The work
\cite{VK-SM:20} proposes a more general moving-horizon estimator via a
perturbed objective function to enable privacy. 
To the best of our knowledge, all of these works only propose sufficient
and theoretical conditions for differential privacy.

However, the design of such algorithms can be subtle and
error-prone. It has been proved that a number of algorithms in the
database literature are incorrect \cite{YC-AM:15} \cite{ML-DS-NL:16},
and their claimed level of privacy can not be achieved. Motivated by
this, the work \cite{ZYD-YXW-GHW-DFZ-DK:18} introduces an approach to
detect the violation of differential privacy for discrete
mechanisms. Yet, this method is only applicable for mechanisms that
result in a small and finite number of events. Further, a precise
characterization of its performance guarantees is not provided.

This motivates our work with contributions in two directions. First,
we build a tractable, data-driven, framework to detect violations of
differential privacy in system estimation.  To handle the infinite
collection of events of continuous spaces, the evaluation is
conditioned over a highly-likely, compact set. This results into a
type of approximate differential privacy with high confidence. We then approximate this set in a data-driven fashion. Further, tests are
performed wrt a collectively-exhaustive and mutually exclusive
partition of the approximated highly-likely set.  By assuming the
probability of these events is upper bounded by a small constant, and
implementing an exact hypothesis test procedure, we are able to
quantify the approximate differential privacy of the estimation wrt
two adjacent inputs with high-likelihood. Second, we employ this
procedure to evaluate the differential privacy of a previously
proposed, $W_2$-MHE estimator. Our experiments show some interesting
results including: i) the theoretical conditions for the $W_2$-MHE
seem to hold but may be rather conservative, ii) there is an
indication that perturbing the output estimation mapping
results in a better performance than perturbing the input sensor
data, iii) differential privacy does depend on sensor locations,
and iv) the $W_2$-MHE performs better than a differentially-private
EKF.

\vspace{-0.1cm}
\section{Problem Formulation} 
\label{sec:Motivation}
Consider a sensor network performing a distributed estimation
 task. Sensors may have different owners, who wish to maintain their
 locations private. Even if communication among sensors is secure, an
 adversary may have access to the estimation on the target, and other
 network side information\footnote{This can be any information
   including, but not limited to, the target's true location, and all
   other sensor positions.}  to gain critical knowledge about any
 individual sensor; see Figure~\ref{fig:Motiv_df}.
   \begin{figure}[t]
   \setlength{\abovecaptionskip}{+0.5cm}
\setlength{\belowcaptionskip}{-0.4cm}
 \centering
 \includegraphics[height=5cm,width=7cm]{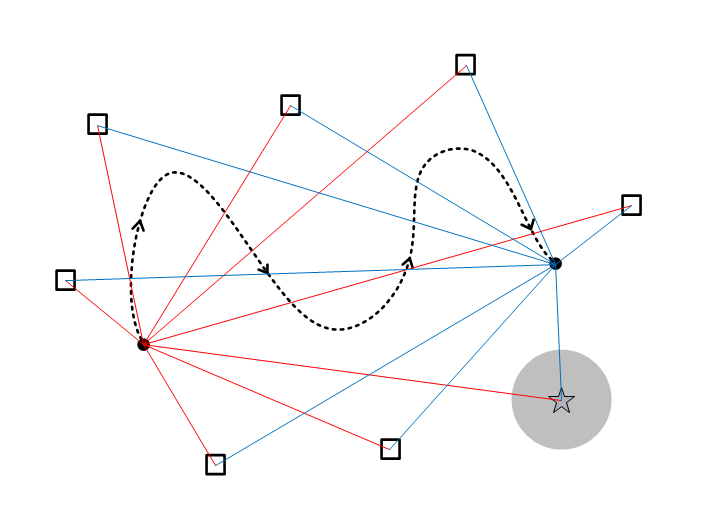}
 \caption{The solid circle represents the moving target that is being
   estimated; the squares represent the location of known sensors (side
   information known by an adversary); the star represents the sensor
   location of a particular sensor that an adversary is tracking. The
   actual location of the starred sensor can be anywhere within the
   shaded circle (hypothesis). The diameter of the hypothesis depends
   on the level of differential privacy of the estimator. The dashed
   curve represents the target's trajectory and the arrows indicate the
   direction. The set of red/blue lines represent the output data
   released from sensors when the target is at the start/end point. In
   practice, an adversary can probably have access to a time history of
   data.}
 \label{fig:Motiv_df}
 \end{figure} 
 
We now start by defining the concept of differential privacy in
estimation, and state our problem objectives.

{\color{black}Let a system and observation models be of the form
}:
\begin{equation}\label{equ:system}
   \Omega: \begin{cases} 
    x_{k+1}=f\left(x_{k}, w_{k}\right), \\
    {\color{black}y_{k}=h\left(x_{k}, v_{k}\right)},
  \end{cases}
\end{equation}
where $x_k \in \mathbb{R}^{d_{X}}, y_k \in
\mathbb{R}^{d_{Y}}$, $w_k \in \mathbb{R}^{d_{W}}$ and $v_k \in
\mathbb{R}^{d_{V}}$. Here, $w_k$ and $v_k$ represent the iid process 
and measurement noises at time step $k$, respectively.


Let $\{0, \ldots, T\}$ be a time horizon, and {\color{black}the sensor data up to time $T$ be denoted by
{\color{black}$\mathrm{y}_{0: T}$ $=\left(y_{0}^{\top}, \ldots,
    y_{T}^{\top}\right)^{\top}$}}. An
estimator or mechanism $\mathcal{M}$ of
{\color{black}\eqref{equ:system}} is a stochastic mapping:
$\real^{(T+1)d_Y}\rightarrow \real^{m d_X}$, for some $m
\ge 1$, which assigns sensor data $\mathrm{y}_{0:T}$ to a random state
trajectory estimate. We will assume that the distribution of
$\mathcal{M}(\mathrm{y}_{0:T})$, $\mathbb{P}$, is independent of the
distribution of $\mathrm{y}_{0:T}$. In Section~\ref{section:
  experiments}, we test a $W_2$-MHE filter that takes this form and
assimilates sensor data online. Roughly speaking, the $W_2$-MHE
employs a moving window $N$ and sensor data $(y_{k+1},\ldots,y_{k+N})$
to estimate the state at time $k$; see~\cite{VK-SM:20} for more
information.
{\color{black} 
  \begin{definition}\textit{(\textbf{($\epsilon$,$d$-adjacent), $\lambda$-approximate, Differential Privacy})}
    Let $\mathcal{M}$ be a state estimator of System~\ref{equ:system}
    and $d_{\mathrm{y}}$ a distance metric on $\real^{(T+1)
      d_Y}$. Given $\epsilon, \lambda, d \in \realnonnegative$,
    $\mathcal{M}$ is ($\epsilon$,
    $d$-adjacent), $\lambda$-approximate,  differentially private if for any 
    $\mathrm{y}_{0: T}^{1}, \mathrm{y}_{0: T}^{2} \in \real^{(T+1)
      d_Y}$, with
    $d_{\mathrm{y}}(\mathrm{y}_{0: T}^{1}, \mathrm{y}_{0: T}^{2}) \le
    d$ we have \vspace*{-0.25cm}
    \begin{equation}\label{equ:event} {\small
    \mathbb{P}(\mathcal{M}\left(\mathrm{y}_{0: T}^{i})
      \in E\right) \leq \mathrm{e}^{\varepsilon}
    {\mathbb{P}(\mathcal{M} (\mathrm{y}_{0: T}^{j})
        \in E))} + \lambda, \; }
  \end{equation}
  for $i, j = 1,2,$ and for all $ E \subset \text
  {range}(\mathcal{M})$. In what follows, we use the notation
  $(\epsilon, d$-adj$)$-$\lambda$ (resp.~$(\epsilon, d$-adj$)$, for
  $\lambda = 0$).
  \end{definition}
  This notion of privacy 
  matches the standard
  definition of~\cite{JC-GED-SH-JLN-SM-GJP:16}\cite{JLN-GJP:14}, for a
  $1$-adjacency relation given by a distance $\le d$.  We are mostly
  interested in the case of $\lambda = 0$ and the standard $(\epsilon,d$-adj$)$ differential privacy \cite{VK-SM:20}. However, later we discuss how to
  choose a $\lambda$ to \textit{approximate} it via $(\epsilon,
  d$-adj$)$-$\lambda$ differential privacy. Ideally, $\lambda$ is to be chosen as
  small as possible.}
Here, $d_{\mathrm{y}}\left(\mathrm{y}_{0: T}^{1},
  \mathrm{y}_{0: T}^{2}\right)$ is induced by the 2-norm. However, the
results do not depend on the choice of metric.

Finding theoretical conditions that guarantee {\color{black}$(\epsilon$,$d$-adj$)$} differential privacy
can be difficult, conservative, and hard to verify. Thus, in this work
we aim to: 
\begin{enumerate}
\item Obtain a tractable, numerical test procedure to evaluate the
   differential privacy of an estimator; while providing quantifiable
   performance guarantees of its correctness.
\item Verify numerically the differential-privacy guarantees of the 
    $W_2$-MHE filter of~\cite{VK-SM:20}; and compare its
   performance with that of an extended Kalman filter.
\item Evaluate the differences in privacy/estimation when the
   perturbations are directly applied to the sensor data before
   the filtering process is done. 
\end{enumerate}

 Our approach employs a statistical, data-driven method. Although a
 main motivation for this work is the evaluation of the $W_2$-MHE
 filter, the method can be used to \textit{verify} the privacy of \textbf{any
 mappings} with a continuous space range.
 \vspace{-0.15cm}
 \section{Approximate {\color{black}$(\epsilon$,$d$-adj$)$}
   Differential Privacy}
\label{section:Approximate}
In this section, we {\color{black} start by introducing a} notion of
{\color{black}high-likelihood $(\epsilon$,$d$-adj$)$ differential
  privacy. This is a first step to simplify the evaluation of
  differential privacy via the proposed numerical framework. A second
  step lies on the identification of a suitable space partition. }


\begin{definition}\longthmtitle{High-likelihood 
    {\color{black}$(\epsilon$,$d$-adj$)$} Differential
    Privacy} \label{def:highlikelyDF}Suppose that $\mathcal{M}$ is a
  state estimator of System~\ref{equ:system}. Given $\epsilon,d \in
  {\color{black}\realnonnegative}$, we say that $\mathcal{M}$ is
  {\color{black}$(\epsilon$,$d$-adj$)$}
  \textit{differentially private with high likelihood $1 -\theta$} 
  if there exists an event $R$ with $\mathbb{P}(R) \ge 1 -
  \theta$ 
  such that, for any two $\mathrm{y}_{0:T}^i$, $i = 1,2$, with
  $d_{\mathrm{y}}(\mathrm{y}_{0:T}^1,\mathrm{y}_{0:T}^2) \le d$,
  we have:
  \begin{equation*}
    \mathbb{P}(\mathcal{M}(\mathrm{y}_{0:T}^i )\in E |R) \le
    e^{\epsilon}  \mathbb{P}(\mathcal{M}(\mathrm{y}_{0:T}^j) \in E |R),
  \end{equation*}
  for $i, j \in\{ 1,2\}$ and all events $E \subseteq
  \text{range}(\mathcal{M})$.
\end{definition}

\begin{lemma}
  Suppose that $\mathcal{M}$ is a high-likelihood {\color{black}$(\epsilon$,$d$-adj$)$}
  differentially private estimator, with likelihood $1 - \theta$. Then,
  $\mathcal{M}$ is  {\color{black}$(\epsilon$,$d$-adj$)$-$\lambda$} differentially
  private {\color{black}with $\lambda = \theta$}.
\end{lemma}
\begin{proof}
  Let $R$ be the high-likely event wrt which $\mathcal{M}$ is
  high-likely differentially private. Let $\overline{R}$ be its
  complement and $E \subseteq \text{range}(\mathcal{M})$. We have
  \begin{align*}
    & \mathbb{P}(\mathcal{M}(\mathrm{y}_{0:T}^1) \in E) \\
    &=
    \mathbb{P}(\mathcal{M}(\mathrm{y}_{0:T}^1) \in E | R) \mathbb{P}(R) +
    \mathbb{P}(\mathcal{M}(\mathrm{y}_{0:T}^1) \in E|\overline{R}) \mathbb{P}(\overline{R})\\
    & \le e^\epsilon \mathbb{P}(\mathcal{M}(\mathrm{y}_{0:T}^2) \in E
    | R) \mathbb{P}(R) + \theta \mathbb{P}(\mathcal{M}(\mathrm{y}_{0:T}^1) \in E |
    \overline{R})\le  \\
     & e^\epsilon [\mathbb{P}(\mathcal{M}(\mathrm{y}_{0:T}^2) \in E
     | R) \mathbb{P}(R) + \mathbb{P}(\mathcal{M}(\mathrm{y}_{0:T}^2) \in E
     | \overline{R}) \mathbb{P}(\overline{R}) ] + \theta 
\\
    & = e^\epsilon  \mathbb{P}(\mathcal{M}(\mathrm{y}_{0:T}^2) \in E)
    + \theta,
  \end{align*}
  similarly, the roles of $\mathrm{y}_{0:T}^1$ and
  $\mathrm{y}_{0:T}^2$ {\color{black}can be exchanged}.  
\end{proof}
Definition~\ref{def:highlikelyDF} still requires checking {\color{black}
  conditions involving an} infinite number of event {\color{black}
  sets}.  Our test framework limits evaluations to a finite
{\color{black}collection as follows.}
\begin{definition}[\bfseries Differential privacy wrt a space
  partition]
  Let $\mathcal{M}$ be an estimator of System \ref{equ:system} and
  $\mathcal{P} = \{E_1, \ldots, E_n\}$ be a space
  partition\footnote{By partition we mean a collection of mutually
    exclusive and collectively exhaustive set of events wrt
    $\mathbb{P}$.} of $\text{range}(\mathcal{M})$. We say that
  $\mathcal{M}$ is {\color{black}$(\epsilon$,$d$-adj$)$} differentially
  private wrt $\mathcal{P}$ if the definition of
  {\color{black}$(\epsilon$,$d$-adj$)$} differential privacy holds
  {\color{black} for} each $E_k \in \mathcal{P}$.
\end{definition}
The following {\color{black} helps} explain the relationship between
{\color{black}$(\epsilon$,$d$-adj$)$} differential privacy wrt a
partition and  {\color{black} the original $(\epsilon$,$d$-adj$)$
  differential privacy}.
\begin{lemma} \label{lemma3}
  Let $\mathcal{M}$ be a state estimator of System \ref{equ:system}, and
  consider a partition of $\text{range}(\mathcal{M})$,
  {\color{black}$\mathcal{P}_1=\{E_1, \ldots, E_{n_1}\}$}, which is \textit{finer} than
  another partition {\color{black}$\mathcal{P}_2$ $=\{F_1, \ldots, F_{n_2}\}$ ($n_1 > n_2$)}.  That is,
  each $F_i$ can be represented by the disjoint union $F_i =
  \union_{s=1}^{m_i}E_{l_s}$.  Then, if $\mathcal{M}$ is
 {\color{black}$(\epsilon$,$d$-adj$)$} differentally private wrt $\mathcal{P}_1$, then it
  is also differentially private wrt $\mathcal{P}_2$.
\end{lemma}

\begin{proof}
By assumption, it holds that
  $\mathbb{P}\left(\mathcal{M}\left(\mathrm{y}_{0: T}^{1}\right) \in
    E_i\right)$ $ \leq \mathrm{e}^{\varepsilon}
  \mathbb{P}\left(\mathcal{M}\left(\mathrm{y}_{0: T}^{2}\right) \in
    E_i\right)$. 
  Take $F_i = E_{l_1} \union \ldots \union E_{l_{m_{i}}}$, then, from
  the properties of partition, we obtain:
\begin{displaymath}\begin{array}{ll}
    &\mathbb{P}\left(\mathcal{M}\left(\mathrm{y}_{0: T}^{1}\right) \in F_i \right)  = \mathbb{P}\left(\mathcal{M}\left(\mathrm{y}_{0: T}^{1}\right) \in E_{l_1} \union \ldots \union E_{l_{m_{i}}}\right)\\
      &
     = \sum^{m_i}_{s=1}\mathbb{P}\left(\mathcal{M}\left(\mathrm{y}_{0: T}^{1}\right) \in E_{l_s} \right) 
    \\ &\le
    e^{\epsilon}\sum^{m_i}_{s=1}\mathbb{P}\left(\mathcal{M}\left(\mathrm{y}_{0:
          T}^{2}\right) 
      \in E_{l_s} \right) 
     = e^{\epsilon}\mathbb{P}\left(\mathcal{M}\left(\mathrm{y}_{0: T}^{2}\right) \in F_i \right).
\end{array}
\end{displaymath}
Similarly, the roles of $\mathrm{y}_{0:T}^1$ and
$\mathrm{y}_{0:T}^2$ {\color{black} can be exchanged.}
\end{proof}
{\color{black}Thus, it follows intuitively that $\mathcal{M}$ is
  $(\epsilon$,$d$-adj$)$ differentially private if it is
  $(\epsilon$,$d$-adj$)$ differentially private wrt infinitesimally
  small partitions}. Now, by considering partitions of a given
\textit{resolution}, we can also guarantee 
{\color{black}a type of approximate $(\epsilon$,$d$-adj$)$}  privacy:
\begin{lemma} 
  Consider a partition $\mathcal{P} = \{E_i\}_{i \in \mathcal{I}}$
  such that $\mathbb{P}(E_i) \le \eta$ for all $i \in \mathcal{I}$.
  Then, if {\color{black}$(\epsilon$,$d$-adj$)$} differential privacy
  holds wrt the partition $\mathcal{P}$, then $\mathcal{M}$ is
  {\color{black}$(\epsilon$,$d$-adj$)$-$\lambda$} differentially
  private with $\lambda = 2\eta e^{\epsilon}$.
\end{lemma}
\begin{proof}
  For any  $R$, we have $\mathbb{P}(\mathcal{M}(\mathrm{y}_{0:T}^1) \in R) 
    = \sum_i \mathbb{P}(\mathcal{M}(\mathrm{y}_{0:T}^1) \in R \cap E_i).$ 
 By hypothesis,
 \vspace{-0.15cm}
  \begin{align*}
    & \mathbb{P}(\mathcal{M}(\mathrm{y}_{0:T}^1) \in E_i) \le e^\epsilon
    \mathbb{P}(\mathcal{M}(\mathrm{y}_{0:T}^2) \in E_i)
  \Longleftrightarrow \\
    & \mathbb{P}(\mathcal{M}(\mathrm{y}_{0:T}^1) \in E_i\cap(R \cup
      \overline{R})) \le e^\epsilon
      \mathbb{P}(\mathcal{M}(\mathrm{y}_{0:T}^2 )\in E_i \cap(R \cup
      \overline{R}) ),
  \end{align*}
  for $i \in \mathcal{I}$, and where $\overline{R}$ is the
  complement of $R$. Thus,
  \begin{align*}
   & \mathbb{P}(\mathcal{M}(\mathrm{y}_{0:T}^1) \in E_i\cap R) +
    \mathbb{P}(\mathcal{M}(\mathrm{y}_{0:T}^1) \in E_i\cap \overline{R}) \\ &\le
    e^\epsilon \mathbb{P}(\mathcal{M}(\mathrm{y}_{0:T}^2) \in E_i \cap R) +
    e^\epsilon \mathbb{P}(\mathcal{M}(\mathrm{y}_{0:T}^2 )\in E_i \cap
    \overline{R}).
  \end{align*}
  Now, if $\mathbb{P}(E_i) \le \eta$, we have
    \begin{align*}
      & \mathbb{P}(\mathcal{M}(\mathrm{y}_{0:T}^1 )\in E_i\cap R) \le
      e^\epsilon \mathbb{P}(\mathcal{M}(\mathrm{y}_{0:T}^2) \in E_i
      \cap R) + 2 \eta e^{\epsilon}.
  \end{align*}
Similarly, the roles of $\mathrm{y}_{0:T}^1,
\mathrm{y}_{0:T}^2$ {\color{black} can be exchanged}. 
\end{proof} {\color{black} Our approach is based on checking
  differential privacy wrt a partition given by $\{\overline{R},
  E_i\}_{i = 1}^n$, where $\overline{R}$ is the complement of a highly
  likely event, and $\{E_i\}$ is a partition of $R$. }

\vspace{-0.2cm}
 \section{Differential Privacy Test Framework}
 \label{section:methodology} 
In this section, we present the components of our test framework (Section~\ref{sec:testFramework}) and its theoretical guarantees (Section~\ref{subsection:guarantees}). 
\vspace{-0.2cm}
\subsection{Overview of the differential-privacy test framework}
\label{sec:testFramework} 
Privacy is evaluated wrt two {\color{black}$d$}-close $\mathrm{y}_{0:T}^1$,
  $\mathrm{y}_{0:T}^2$ as follows: 
\begin{enumerate}

\item 
  Instead of verifying {\color{black}$(\epsilon$,$d$-adj$)$ privacy}
  for an infinite number of events, an \texttt{\small
    EventListGenerator} module extracts a finite collection
  \texttt{\small EventList}. This is done by partitioning an
  approximated high-likely event set.

\item Next, a \texttt{\small WorstEventSelector} module identifies the
  worst-case event in \texttt{\small EventList} that violates
  {\color{black}$(\epsilon$,$d$-adj$)$} differential privacy with the
  highest probability.

\item Finally, a \texttt{\small HypothesisTest} module evaluates {\color{black}$(\epsilon$,$d$-adj$)$} ifferential privacy wrt the worst-case event.

\end{enumerate}
The overall description is summarized in Algorithm~\ref{alg: overview}.
\begin{algorithm}[h]
\begin{algorithmic}[1]
    \Function{Test Framework}{$\mathcal{M}, \epsilon$, $\mathrm{y}_{0: T}^{1}, \mathrm{y}_{0: T}^{2}$}
    \State \textbf{Inputs:} Target estimator $\mathcal{M}$, privacy
    level $\epsilon$, sensor data
    ($\mathrm{y}_{0: T}^{1}, \mathrm{y}_{0: T}^{2}$)
    \State \texttt{\small EventList} = \texttt{\small EventListGenerator}($\mathcal{M}$, $\mathrm{y}_{0: T}^{1}$)
    \State \texttt{\small WorstEvent} = \State \texttt{\small WorstEventSelector}($\mathcal{M}, \epsilon, \mathrm{y}_{0: T}^{1}, \mathrm{y}_{0: T}^{2}$, 
    \State $\quad $ \texttt{\small EventList})
    \State $p^{+},p_{+}$ = \texttt{\small HypothesisTest}($\mathcal{M}, \epsilon$, $\mathrm{y}_{0: T}^{1}, \mathrm{y}_{0: T}^{2}$, 
    \State $\quad$  \texttt{\small WorstEvent})
    \State Return $p^{+},p_{+}$
    \EndFunction
  \end{algorithmic}
\caption{{\color{black}$(\epsilon$,$d$-adj$)$} Differentially-private Test Framework}
\label{alg: overview}
\end{algorithm} 

We now describe each module in detail. 
\subsubsection{\texttt{\small
    EventListGenerator} module.} \label{sec:EventListGenerator}

Consider System~\ref{equ:system}, with initial condition $x_0 \in
\mathcal{K}_0 \subset \mathbb{R}^{d_{X}}$. 
The estimated state $\hat{x}_k$ under $\mathcal{M}$ for a given
$\mathrm{y}_{0:T}^1$, belongs to a set of all possible
estimates 
given $x_0$ and $\mathrm{y}_{0:T}^1$. Denote this set as $ R_{[0, k]} $.



In~\cite{VK-SM:20}, this set is bounded as all disturbances and
initial distribution are assumed to have a compact support. However,
$R_{[0, k]}$ can be unbounded for other estimators. To reduce the set
of events to be checked for {\color{black}$(\epsilon$,$d$-adj$)$} differential privacy:
a) we approximate the set $R_{[0, k]}$ by a compact,
\textit{high-likely set} in a data-driven fashion, and b) we finitely
partition this set by a mutually exclusive collection of events; see
Algorithm~\ref{alg: Generator}.
\vspace{-0.8cm}
\setlength{\textfloatsep}{-0.05cm}
\begin{algorithm}[h] 
\begin{algorithmic}[1]
  \Function{EventListGenerator}{$\mathcal{M}$, $\mathrm{y}_{0:
      T}^{1}$, {\color{black}$\beta, \gamma$}} \State \textbf{Input:} Target Estimator($\mathcal{M}$)
  \State \phantom{In put:} Sensor Data($\mathrm{y}_{0: T}^{1}$)
  \State \phantom{In put:} {\color{black}Parameters for Algorithm~\ref{alg: Reachable} ($\beta, \gamma$)}
  \State \texttt{\small HighLikelySet} $\leftarrow$ {\color{black}Apply Algorithm~\ref{alg: Reachable}}
  \State \texttt{\small EventList}
  $\leftarrow$ {\color{black} a partition 
   of the \texttt{\small HighLikelySet}}
 \State Return \texttt{\small EventList}
    \EndFunction
  \end{algorithmic}
\caption{\texttt{\small EventListGenerator}}
\label{alg: Generator}
\end{algorithm} 

\vspace{-0.8cm}
Inspired by~\cite{AD-MA:20} focusing on reachability,
we employ the Scenario Optimation approach to approximate a
high-likely set via a product of ellipsoids; see Algorithm~\ref{alg:
  Reachable}. Here, $\Gamma \equiv \Gamma(\beta)$ defines the number
of estimate samples (filter runs) required to guarantee that the
output set contains $1 - \beta$ of the probability mass of $R_{[0,k]}$
with high confidence $1
-\gamma$. 
Then, a convex optimization problem is solved to find the output set 
as a hyper-ellipsoid with parameters $A^k$ and $b^k$. 
\begin{algorithm}[h] 
\begin{algorithmic}[1]

\State \textbf{Input:} Target Estimator($\mathcal{M}$) with dimension
$d_X$ 
\State \phantom{In put:} Sensor data($\mathrm{y}_{0: T}^{1}$),
parameters  $\beta, \gamma$
\State \textbf{Output:} Matrix $A^{k}$ and vector $b^{k}$ representing
an \State \phantom{Ou tput:}$1$-$\beta$-accurate high-likely set at time step $k$ \State
\phantom{Ou tput:} $R_{k}(A^k,b^k)=\left\{x \in \real^{d_X}\,|\,\|A^k
  x+b^k\|_{2} \leq 1\right\}$ \State \phantom{Ou tput:} with
confidence $1 - \gamma$.

\State Set number of samples $\Gamma$ = 
\State \phantom{Ou tput:}$\left[\frac{1}{\beta}
 \frac{e}{e-1}\left(\log \frac{1}{\gamma}+d_X(d_X+1) / 2+d_X\right)\right\rceil$
\For{$k \in \{0,\ldots,T\}$}
\For{$i \in \{0,\ldots,\Gamma\}$}
\State Record $z^k_i$ = $\mathcal{M}\left(\mathrm{y}_{0: k}^{1}\right)$
\EndFor
\State Solve the convex problem 
\State
$\begin{array}{ll}\arg \min _{A^k, b^k} & -\log \operatorname{det} A^k
  \\ \text {subject to } & \left\|A^k z^{k}_i-b^k\right\|_{2}-1 \leq
  0,\,i=0,\ldots,\Gamma \end{array}$ 
\State return
$A^k, b^k$
\EndFor
\end{algorithmic}
\caption{\texttt{\small HighLikelySet} }
\label{alg: Reachable}
\end{algorithm} 

The data-driven approximation of the high-likely set can now be
partitioned using e.g.~a grid per time step. This is what we do in
simulation later. {\color{black}For the $W_2$-MHE, the last $N$ (window
  size) steps would not be evaluated for differential privacy. Thus,
  the high-likely set is the product of $T+1-N$ ellipsoids.}  
  Alternatively, an outline on how to re-use the sample runs of Algorithm~\ref{alg: Reachable} to {\color{black}find a finite
  partition} of the approximated set with a common upper-bounded
probability is provided as follows.

Observe that $\Gamma \equiv \Gamma(\beta)$ is a function of $\beta$,
and denote the output of Algorithm~\ref{alg: Reachable} by $R$. For
$\eta > \beta$, we can choose a number of samples $\Gamma (\eta) <
\Gamma(\beta)$ and solve a similar convex problem as in step {\small
  [13:]}. The resulting hyper-ellipsoid, $\overline{E}$, satisfies
$\overline{E} \cap R \neq \emptyset$ as they contain common sample
runs, and $\mathbb{P}(\overline{E}\cap R) = \mathbb{P}(\overline{E}) +
\mathbb{P}(R) - \mathbb{P}(\overline{E} \cup R) \ge 1 - \eta + 1 -
\beta - 1 = 1 -\eta -\beta$ with confidence $1 - \gamma$. Similarly,
the complement of $\overline{E}$, $E$, is such that $E \cap R \neq
\emptyset$ and $\mathbb{P}(E \cap R) = \mathbb{P}(R) - \mathbb{P}(R
\cap \overline{E}) \le 1 - (1 - \eta -\beta) = \eta + \beta$. This
process can be repeated by choosing different subsets of sample runs
to find a finite collection of sets $E_1 \cap R, \dots, E_n \cap R$
such that $\mathbb{P}(E_i \cap R) \le \eta + \beta$ and $(E_1 \cap R)
\cup \dots \cup (E_n \cap R) = R$. Wlog, it can be assumed that the
sets are mutually exclusive by re-assigning sets overlaps to one of
the sets. This approach can be extended to achieve any desired upper
bound $\eta_d$ for a desired $\beta_d$. To do this, first run
Algorithm~\ref{alg: Reachable} with respect to $\bar{\beta}$ so that $
\eta_d > 2 \bar{\beta} $. This results into a set $\bar{R}$ with
probability at least $1 - \bar{\beta}$ and high confidence $1 -
\gamma$. By selecting a subset of sample runs $\Gamma(\beta_d)$ and
solving the associated optimization problem, one can obtain $R$ with
the desired probability lower-bound $1 - \beta_d$. Now, following the
previous strategy for $\bar{\eta} = \eta_d- \bar{\beta}$, we can
obtain a partition of $R$ and $\bar{R}$ with probability upper bounded
by $\bar{\eta} + \bar{\beta} = \eta_d$ and high confidence
$1-\gamma$. 

    
We also note that a finite partition of a compact set is
always guaranteed to exist under certain conditions, in particular, as the following:
\begin{lemma}
Let $R$ be a compact set in $\real^d$. If $\mathbb{P}$ is an absolutely continuous measure wrt the Lebesgue measure in $\real^d$, and if the Randon-Nikodym derivative of $\mathbb{P}$ is a continuous function, then
there is a finite partition of $R$ of a given resolution $\eta$. 

\end{lemma}

\begin{proof}
    First of all, from absolutely continuity, $\mathbb{P}(E) = \int_E
    f d\mu$, where $\mu$ is the Lebesgue measure, and $f$ the
    Randon-Nikodym derivative. Second, observe that we can take
    arbitrarily small-volume neighborhoods $E_x$, of points of $ x \in
    R$, to make $\mathbb{P}(E_x \cap R)$ as small as we like. This
    follows from $\mathbb{P}(E_x \cap R) \le \text{vol}(E_x)
    \sup_{E_x\cap R} f \le \text{vol}(E_x) \max_R f \le \eta$, where
    $\text{vol}(E_x)$ is the standard volume of the set $E_x$ wrt the
    Lebesgue measure, and $\max_R f <\infty$ by compactness of $R$ and
    continuity of $f$. Using these neighborhoods, we can construct an
    open cover $\{ E_x\}_{x\in R}$ of $R$, i.e.~$R \subseteq \cup_{x
      \in R} E_x$. By compactness of $R$, there must exist a finite
    subcover $\{E_{1} ,\dots, E_{n}\}$, i.e. $R \subseteq E_1 \cup \dots
    \cup E_n$. From this cover, we can obtain
    $\{E_1\cap R, \dots, E_n\cap R\}$, with probabilities
    $\mathbb{P}(E_i \cap R) \le \eta$, for each $i$. Now the sets $E_i
    \cap R$ can then be used to construct a finite partition of $R$,
    with probabilities that will be less or equal than $\eta$.
\end{proof}

\subsubsection{\texttt{\small WorstEventSelector} module.}
We now discuss how to select an $E^{*}$, a most-likely event that
leads to a violation of {\color{black}$(\epsilon$,$d$-adj$)$}
differential privacy; see Algorithm~\ref{alg: Selector}.
\begin{algorithm}[h] 
\begin{algorithmic}[1]
  \Function{WorstEventSelector}{$n, \mathcal{M}$, $\epsilon$,
    $\mathrm{y}_{0: T}^{1}, \mathrm{y}_{0: T}^{2}$, \texttt{\small
      EventList}} \State \textbf{Input:} Target
  Estimator($\mathcal{M}$) \State \phantom{In put:} Desired
  differential privacy($\epsilon$) \State \phantom{In put:}
  {\color{black}$d$}-adjacent sensor data($\mathrm{y}_{0: T}^{1},
  \mathrm{y}_{0: T}^{2}$) \State \phantom{In put:} \texttt{\small EventList}
    
    \State $O_1$ $\leftarrow$ Estimate set  after $n$ runs of $\mathcal{M}(\mathrm{y}_{0: T}^{1})$
    \State $O_2$ $\leftarrow$ Estimate  set after $n$ runs of $\mathcal{M}(\mathrm{y}_{0: T}^{2})$ 
    \State $\texttt{\small pvalues} \leftarrow [\phantom{i}]$
    \For{$E\in$ \texttt{\small EventList}}
    \State $c_1$ $\leftarrow$ $|\{i|O_1[i] \in E\}|$
    \State $c_2$ $\leftarrow$ $|\{i|O_2[i] \in E\}|$
    \State $p^{+},p_{+} \leftarrow$ PVALUE $(c_1,c_2,n,\epsilon)$
    \State $p^{*} \leftarrow \min{(p^{+},p_{+})}$
    \State $\texttt{\small pvalues}$.append($p^{*}$)
    \EndFor
    \State \texttt{\small WorstEvent}  $\leftarrow$ \texttt{\small
      EventList}[$\argmin\{\texttt{\small pvalues}\}$] \State Return
    {\small $E^* =$} \texttt{\small WorstEvent}
    \EndFunction
  \end{algorithmic}
\caption{\texttt{\small WorstEvent} Selector}
\label{alg: Selector}
\end{algorithm} 
The returned \texttt{\small WorstEvent} ($E^*$) is then used in
Hypothesis Test function in Algorithm~\ref{alg: overview}. First,
\texttt{\small WorstEventSelector} receives an \texttt{\small
  EventList} from \texttt{\small EventListGenerator}. The algorithm
runs the estimator $n$ times with each sensor data, and counts the
number of estimates that fall in each event of the list,
respectively. Then, a \textbf{PVALUE} function is run to get $p^{+},
p_{+}$; 
see Algorithm~\ref{alg: Hypothesis} (top). The $p$-values
quantify the probability of the Type~I error of the test (refer to the next
subsection).

\subsubsection{\texttt{\small HypothesisTest} module.}
The hypothesis test module aims to verify
{\color{black}$(\epsilon$,$d$-adj$)$ privacy} in a data-driven fashion
for a fixed $E^*$.  Given a sequence of {\color{black}$m$}
statistically-independent trials, the \textit{ocurrence} or
\textit{non-occurrence} of $\mathcal{M}(\mathrm{y}_{0:T}^i) \in E$ is
a Bernouilli sequence and the number of occurrences in
{\color{black}$m$} trials is distributed as a Binomial.
\begin{algorithm}[t] 
\begin{algorithmic}[1]
    \Function{\text{pvalue}}{$c_1, c_2, n, \epsilon$}
    \State $\bar{c_1} \leftarrow$ B($c_1, 1/e^{\epsilon}$)
    \State $s \leftarrow \bar{c_1} + c_2$
    \State $p^{+} \leftarrow$ 1 - Hypergeom.cdf($\bar{c_1}-1|2n,n,s$) \State $\bar{c_2} \leftarrow$ B($c_2, 1/e^{\epsilon}$)
    \State $s \leftarrow \bar{c_2} + c_1$
    \State $p_{+} \leftarrow$ 1 - Hypergeom.cdf($\bar{c_2}-1|2n,n,s$)
    \State return $p^{+}, p_{+}$
    \EndFunction
    \Function{HypothesisTest}{${\color{black}m}, \mathcal{M}$, $\epsilon$, $\mathrm{y}_{0: T}^{1}, \mathrm{y}_{0: T}^{2}$, $E^*$}
    \State \textbf{Input:} Target Estimator($\mathcal{M}$)
    \State \phantom{In put:} Desired differential privacy($\epsilon$)
    \State \phantom{In put:} {\color{black}$d$}-adjacent sensor data($\mathrm{y}_{0: T}^{1}, \mathrm{y}_{0: T}^{2}$)
    \State \phantom{In put:} $E^*$(\texttt{\small WorstEvent})
     \State $O_1$ $\leftarrow$ Estimate set  after $m$ runs of $\mathcal{M}(\mathrm{y}_{0: T}^{1})$
    \State $O_2$ $\leftarrow$ Estimate set  after $m$ runs of $\mathcal{M}(\mathrm{y}_{0: T}^{2})$
    \State $c_1$ $\leftarrow$ $|\{i|O_1[i] \in E^*\}|$
    \State $c_2$ $\leftarrow$ $|\{i|O_2[i] \in E^*\}|$
    \State $p^{+},p_{+} \leftarrow$ PVALUE $(c_1,c_2,m,\epsilon)$
    \State Return $p^{+},p_{+}$
    \EndFunction
  \end{algorithmic}
\caption{\texttt{\small HypothesisTest}}
\label{alg: Hypothesis}
\end{algorithm}

Thus, {\color{black}the verification} reduces to evaluating how the
parameters of two Binomial distributions differ. More precisely,
define $p_i = \mathbb{P}\left(\mathcal{M}\left(\mathrm{y}_{0:
      T}^{i}\right) \in E^{*}\right)$, $i = 1,2$. 
By running {\color{black}$m$} times the estimator, we count the
number of $\mathcal{M}\left(\mathrm{y}_{0: T}^{i}\right) \in E^{*}$ as
$c_i$, $i = 1,2$.  Thus, each $c_i$ can be seen as a sample of the
binomial distribution B$(n_i,p_i)$, for $i = 1,2$.  However, instead
of evaluating $p_1 \le p_2$, we are interested in testing the null
hypothesis $p_1 \leq \mathrm{e}^{\varepsilon} p_2$, with an additional
$\mathrm{e}^{\varepsilon}$. This can be addressed by considering
samples $\bar{c}_1$ of B($c_1, 1/e^{\epsilon}$) distribution. It is easy to
see the following:
\begin{lemma}[\cite{ZYD-YXW-GHW-DFZ-DK:18}.] \label{lemma:frompaper}
  Let $Y \sim$ B($n,p_1$), and $Z$ be sampled from B($Y,
  1/e^{\epsilon}$), then, $Z$ is distributed as B($n,p_1/e^{\epsilon}$). 
\end{lemma}

Hence, the problem can be reduced to the problem of testing the null hypothesis
$\bar{H}_0 :\bar{p}_1:= p_1/e^{\epsilon} \leq p_2$ on the basis of the
samples $\bar{c}_1$, $c_2$.  Checking whether or not $\bar{c}_1$,
$c_2$ are generated from the same binomial distribution can be done
via a Fisher's exact test \cite{RAF:35} wth $p$-value being equal to
1 -
Hypergeometric.cdf{\color{black}($\bar{c}_1-1|2m,m,\bar{c}_1+c_2$)} (cumulative hypergeometric distribution)
  . As
an exact test, $c_1 \gg e^\epsilon c_2$ provides firm evidence against
the null hypothesis with a given
confidence.

The $p$-value is the Type~I error of the test, or the probability
of incorrectly rejecting $H_0$ when it is indeed
true. 
The null hypothesis is rejected based on a significance level $\alpha$
(typically $0.1$, $0.05$ or $0.01$).  If $p$ is such that $p \le
\alpha$, then the probability that a mistake is made by rejecting
$H_0$ is very small (smaller or equal than $\alpha$). 
Since Condition~\ref{equ:event} involves two inequalities $p_1 \le
e^\epsilon p_2$ and $p_2 \le e^\epsilon p_1$, we need to evaluate two
null hypotheses. This results into $p_+$ and $p^+$ values that should
be larger than $\alpha$ if we want to accept $H_0$; see
Algorithm~\ref{alg: Hypothesis}. In simulation, we choose
  $m > n$ in both \texttt{\small WorstEventSelector} and
  \texttt{\small HypothesisTest} for the purpose of i) increasing the
  accuracy of the test for the worst event, ii) some additional
  practical considerations as follows. 
  
  In the \texttt{\small WorstEventSelector}, the
    $p$-values for all the events are computed using the same
    $\epsilon$, for example, $\epsilon=1$. And then the
    worse event $E^{*}$ which gives the minimum $p$-value is selected. Later, in
    \texttt{\small HypothesisTest}, $c_1, c_2$ are counted with
    respect to the worse event and then we keep increasing $\epsilon$ and
    running \textbf{PVALUE} with the same $c_1, c_2$ until the
    $p$-value is larger than $\alpha$. The critical values are
    reported in Section~\label{section: experiments}. The
    reasons why we select the same $\epsilon$ in \texttt{\small
      WorstEventSelector} are:
      \begin{itemize}
\item The p-value evolves consistently with respect to the test
  $\epsilon$. For example, the event with $c_1 = 12, c_2 = 27$ is the
  most possible event that shows the violatation of differential
  privacy. From  Figure~\ref{fig:test1}, we can see that no matter
  what value $\epsilon$ is chosen, the $p$-values for that event are
  the smallest, which means in \texttt{\small WorstEventSelector}, it
  will always be returned as the worst event $E^{*}$. This happens
  similarly for the event selected from $c_1=14, c_2=27$, the
  $p-$-values are always the largest. In that case, it is safe to
  select one specific $\epsilon$ in \texttt{\small WorstEventSelector}
  to decide which event is the worst.
  \begin{figure}[H]
\centering
  \includegraphics[width=0.8\linewidth]{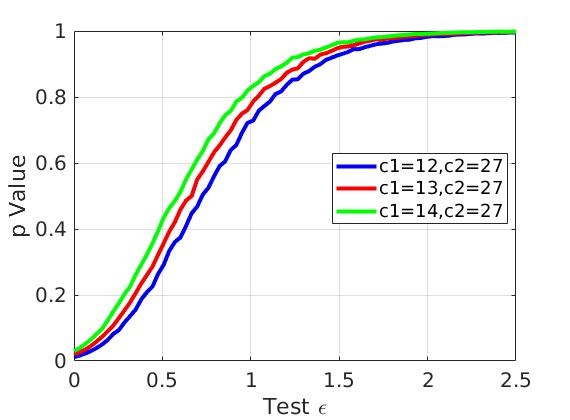}
  \caption{$p$-values evolution with respect to different choices of event. It is clear that no matter what value $\epsilon$ is chosen, $p$-values for the worst event are the smallest.}
  \label{fig:test1}
\end{figure}
\item On the other hand, the variations may come from the
  scenarios that $c_2/c_1=$ constant, but $(c_1,c_2)$ share different
  values, such as (10,20), (15,30), (20, 40). From
  Figure~\ref{fig:test2}, for different $\epsilon$, the $p$-values are
  almost similar. Hence, in this case, in \texttt{\small
    WorstEventSelector}, the returned worst event is still the one of
  the events, if any, that are most likely to show the violations of
  differential privacy.
\begin{figure}[H]
\centering
  \includegraphics[width=0.8\linewidth]{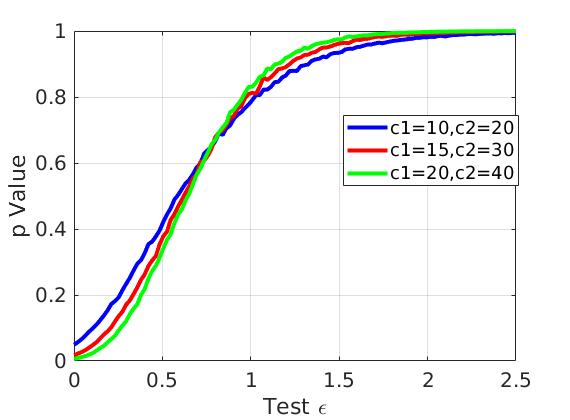}
  \caption{$p$-values evolution with respect to different choices of event. For different choices of event, the variations are not significant.}
  \label{fig:test2}
\end{figure}
\end{itemize}

\vspace{-0.25cm}
\subsection{Theoretical Guarantee}\label{subsection:guarantees}
Here, we specify the guarantees of the numerical test.
\begin{theorem}
  Let ${\mathcal M}$ be a state estimator of System~\ref{equ:system}, and
  let $\epsilon, d, \beta$ and $\gamma \in
  \realnonnegative$. {\color{black}We denote two {\color{black}$d$}-adjacent sensor data as $\mathrm{y}_{0:T}^i$, $i \in\{1,2\}$ and a partition of the high-likely ($1 - \beta$) set $R$  from Algorithm~\ref{alg: Reachable} with high confidence $1 -
  \gamma$ as $\mathcal{P} =\{E_1, \ldots, E_{n}\}$ such that $\mathbb{P}(E_i) \le \eta$ for all $i$}
  . 
  Then, if $\Gamma$ is selected accordingly, and the estimator passes
  the test in Algorithm~\ref{alg: overview}, then 
  ${\mathcal M}$ is approximately {\color{black}$(\epsilon$,$d$-adj$)$}
  differentially private wrt $\mathrm{y}_{0:T}^i$, $i \in\{1,2\}$, and
  $\lambda = \beta + 2 \eta \text{e}^\epsilon$, with confidence $(1 -
  \alpha)( 1 - \gamma)$. 
\end{theorem}
\begin{proof}
  First, from Scenario Optimization $R$ is a highly likely event
  ($1-\beta$) with high confidence ($1-\gamma$), and approximate
  differential privacy with $\lambda = \beta + 2 \eta e^{\epsilon}$ is
  a consequence of the previous lemmas {\color{black}in Section~\ref{section:Approximate}}. Second, as Fisher's test is
  exact, the chosen significance level $\alpha$ characterizes exactly
  the Type I error of the test for any number of samples. If Fisher's
  test is applied to the worst-case event in the partition, we can
  guarantee approximate differential privacy with confidence $1 -
  \gamma$ does not hold with probability $\alpha$. 
\end{proof}
To extend the result over any two {\color{black}$d$}-adjacent sensor data, one would have to sample over the measurement space and evaluate the ratio of
{\color{black} passing the tests.}
\vspace{-0.2cm}
\section{EXPERIMENTS} \label{section: experiments} 
In this section, we evaluate our test on a toy dynamical system. All
simulations are performed in MATLAB (R2020a). 

\paragraph{System Example.} Consider a non-isotropic oscillator in
$\real^2$
  with potential function $ V\left(x^{1}, x^{2}\right)=\frac{1}{2}((x^{1})^{2}+4(x^{2})^{2}).$ 
  Thus, the corresponding oscillator particle with position
  $\mathbf{x}_k=(x^1_k,x^2_k) \in \real^2$ moves from initial
  conditions $x^1_0=5, x^2_0=0, \dot{x}^{1}_0=0, \dot{x}^{2}_0=2.5$
  under the force $-\Delta V$. The discrete update equations of our
  (noiseless) dynamic system take the form
 $   \mathbf{x}_{k+1} = f_0(\mathbf{x}_k) = A\mathbf{x}_k, \, k \ge 0$,
where, $A$ is a constant matrix. 
At each time step $k$, the system state is perturbed by a uniform
distribution over $[-0.001,0.001]^2$.
The distribution of initial conditions is given by a truncated
Gaussian mixture
with mean
vector $(5, 0, 0, 2.5)$ 
uch that $\operatorname{diam}(\mathcal{K}_{0})$ in
Theorem~4 of \cite{VK-SM:20} is~0.1.

The target is tracked by a sensor network of 10 nodes
located on a circle with center $(0,0)$ and radius $10\sqrt{2}$. The
sensor model is homogeneous and given by:
\begin{align*}
  \mathrm{y}_{k}^i& = h(\mathbf{x}_k, \mathbf{q}_i) + \mathbf{v}_{i,k} \\
  & =
  100 \tanh \left(0.1
    \left(\mathbf{x}_k-\mathbf{q}_{i}\right)\right)+\mathrm{v}_{k}^i, \quad i=1,
  \ldots, 10,
\end{align*}
where $\mathbf{q}_{i} \in \real^2$ is the position of sensor $i$ on
the circle, and $\mathbf{x}_k=\left(x^{1}_k, x^{2}_k\right)$ is the
position of the target at time  $k$. Here, the hyperbolic
tangent function $\tanh$ is applied in {\color{black}an} element-wise way. 
The vector $\mathrm{v}_{k}^i\in \real^2$ represents the observation
noise of each sensor, which is generated from the same truncated
Gaussian mixture distribution at each time step $k$ in simulation,
indicating that the volume of random noise has a limit.
All these observations are stacked together as sensor data
$\mathrm{y}_k=(\mathrm{y}_{k}^{1 \top},\dots,
\mathrm{y}_{k}^{10\top})^\top \in \real^{20} $
  in $W_2$-MHE.

  In order to implement the $W_2$-MHE filter, we consider a time
  horizon $T = 8$ and a moving horizon $N = 5 \le T$. 
  
  With these simulation settings, $c_f, l$ in Theorem~4 of
  \cite{VK-SM:20} can be computed as: $c_f=\|A\|$, $l=2(N+1)
  100\|A\|$.

\paragraph{{\color{black}$d$}-adjacent sensor data.} 
Let us use $\mathbf{\theta}^1$ and $\mathbf{\theta}^2 \in \real$ to
represent the angle of a single sensor that is moved to check for
differential privacy. Denote $\Delta \theta = \| \theta^1 - \theta^2
\|$ (distance on the unit circle). 
By exploiting the Lipschitz's properties of the function $h$, it can
be verified that the corresponding measurements satisfy:
$
  d_{\mathrm{y}}
  \left(\mathrm{y}_{0: T}^{1}, \mathrm{y}_{0: T}^{2}\right) \leq 10 \sqrt{2(T-N+1)} \Delta \mathbf{\theta}  =20  \sqrt{2}   \Delta \mathbf{\theta}.
$
 Thus, 
in order to generate {\color{black}$d$}-adjacent sensor data, we take
$
\Delta \mathbf{\theta} \leq \frac{{\color{black}d}}{20  \sqrt{2} }.
$
In the sequel, we take ${\color{black}d} = 10$. 
\newline

\vspace*{-1ex}

\textbf{Numerical Verification Results of $W_2$-MHE.}
Here, evaluate the differential privacy of $W_2$-MHE with parameters
$T=8$ and $N = 5$.  An entropy factor $s_k \in [0,1]$ determines the
distribution of the filter, from $s_k = 1, \forall k$ --- a deterministic
$\mathcal{M}$ to $s_k=0, \forall k$ --- a uniformly distributed random variable.

Fixing $\mathrm{y}_{0:T}^1$, we run the $W_2$-MHE filter for a number
of $\Gamma$ ($= 814$) runs to obtain a high-likely set characterized
by $\beta=0.05, \gamma=10^{-9}$.  
This allows us to produce an ellipsoid that contains at least 0.95 of
the high-likely set with probability $\ge 1 - 10^{-9}$.  At each time
step $k$, we consider a grid partition of the high-likely set
consisting of 4 regions ($r=2$).  The \texttt{\small EventList} is
obtained by storing all of the possible combinations of these sets,
which results in a total of $4^4 = 256$ for $T = 8, N = 5$. Followed
by this, we re-run $W_2$-MHE enough times using both sets of sensor
data, respectively. As shown in Algorithm~\ref{alg: Selector}, we
obtain $c_1, c_2$ for each event and, finally the event with minimum
$p$-value is returned as the \texttt{\small WorstEvent}.  After this,
we record $c_1, c_2$ with respect to this event from another set of
runs. Then, $p$-values are computed for different values of
$\epsilon$. We then use these with the significance parameter
$\alpha$ {\color{black}(0.05 in this work)} to accept or reject the null hypothesis and decide whether {\color{black}$(\epsilon$,$d$-adj$)$} differential privacy is satisfied.

For a fixed sensor setup $\mathbf{Q_1}$, we now report on the
numerical results. 
  We first show the simulation results when no
entropy term exists ($s_k=1, \forall k$ in Theorem~4 of
\cite{VK-SM:20}. We use $s$ to denote $s_k, \forall k$ below).
In Figure~\ref{fig:W2MHE1}, from the left figure, the $p$-value is always equal to 0, which means that the
null hypothesis should always be rejected for all test
$\epsilon$. Thus, the two sets of sensor data are distinguishable when
$s=1$.
On the other hand, the estimate RMSE error using the correct sensor data
generated from $\mathbf{Q_1}$ ($\subscr{E}{correct}$) is almost zero
and much smaller than $\subscr{E}{adjacent}$. This error ($\subscr{E}{adjacent}$) employs
sensor data that are in fact generated from adjacent sensor
positions to $\mathbf{Q_1}$.
\begin{figure}[t]
\centering
\begin{subfigure}{0.23\textwidth}
\centering
\includegraphics[height=3.6cm,width=4.2cm]{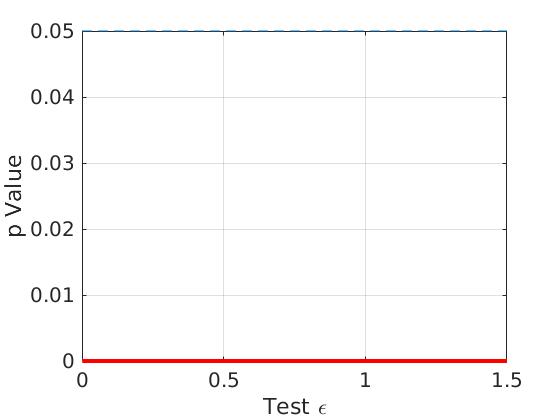}
\caption{Hypothesis test results}
\label{fig:epsilon08}
\end{subfigure}
\begin{subfigure}{0.23\textwidth}
\centering
\includegraphics[height=3.6cm,width=4.2cm]{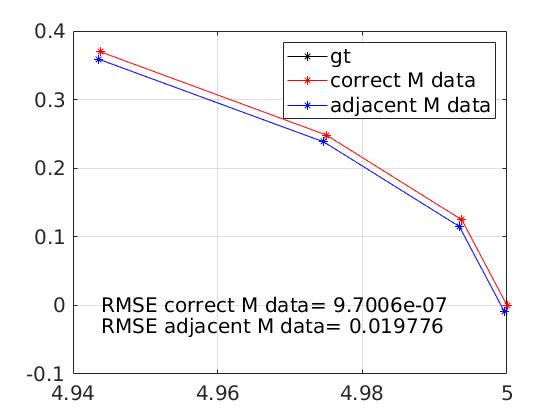}
\caption{Estimation accuracy}
\end{subfigure}
\caption{$W_2$-MHE: State estimation \& privacy test results ($s =
  1$)}
\label{fig:W2MHE1}
\end{figure}

Then, we set $s=0.8$ and the simulation results are shown in Figure~\ref{fig:W2MHE} (left figure). 
We obtain the critical value $\epsilon_c=0.39947$, where the $p$-value
is larger than 0.05. This means that the null hypothesis is not
rejected and we accept that differential privacy holds for $\epsilon
\ge 0.39947$, and $d = 10$ and these two sensor data sets. The
corresponding estimate errors are found to be $\subscr{E}{{correct}} =
0.0040408$ and $\subscr{E}{\text{adjacent}} = 0.026032$ for the
estimates using the sensor data generated from $\mathbf{Q_1}$  and adjacent sensor positions,
respectively. Recall that $s=0.8$
implies a relatively low noise injection level. 
\begin{figure}[t]
\centering
\begin{subfigure}{0.23\textwidth}
\centering
\includegraphics[height=3.6cm,width=4.2cm]{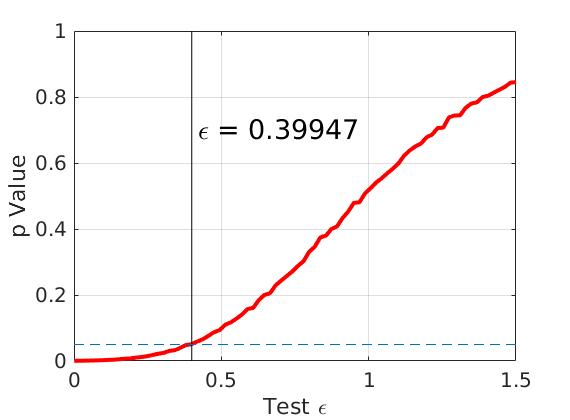}
\caption{$s=0.8$}
\label{fig:epsilon08}
\end{subfigure}
\begin{subfigure}{0.23\textwidth}
\centering
\includegraphics[height=3.6cm,width=4.2cm]{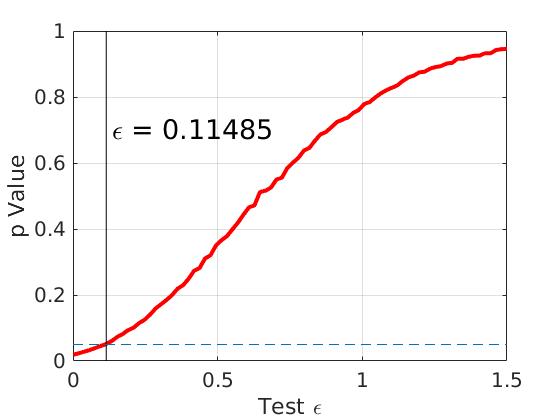}
\caption{$s=0.7$}
\end{subfigure}
\caption{$W_2$-MHE: Privacy test results}
\label{fig:W2MHE}
\end{figure}

Decreasing $s$ to 0.7, leads to a larger entropy term in the
$W_2$-MHE: as shown in Figure~\ref{fig:W2MHE} (right figure), 
the critical
$\epsilon_c$ becomes smaller ($\epsilon_c = 0.11485$), confirming that a
higher level of {\color{black}$(\epsilon$,$d$-adj$)$} differential privacy is achieved. {\color{black}There is also} a
decrease in accuracy, which can be seen from $E_{\text{correct}} =
0.00599996$.

Therefore, the tests reflect the expected trade-off between
differential privacy and accuracy. In order to choose between two
given estimation methods, a designer can either (i) first set a bound
on what is the tolerable estimation error, then compare two methods
based on the differential privacy level they guarantee based on the
given test, or (ii) given a desired level of differential privacy,
choose the estimation method that results into the smallest estimation
error.

{\color{black}Regarding the approximation term $\lambda$, we can compute their values as discussed in Section~\ref{subsection:guarantees}. For each event $E_i$, $\mathbb{P}(E_i)$ can be approximated as $c_1/n$, so $\eta$ is obtained via $\max(\mathbb{P}(E_i))$; $\beta=0.05$ is fixed and $\epsilon_c$ values are known from tests. Therefore, when $s=0.8$, we get $\eta=0.013$ and $\lambda=0.0888$; $s=0.7$, we get $\eta=0.010$ and $\lambda=0.0724$.}
 \paragraph{Input Perturbation. }
 The work \cite{JLN-GJP:14} proposes two approaches to obtain
 differentially-private estimators. The first one randomizes the
 output of regular estimator ($W_2$-MHE belongs to this class). The
 second one perturbs the sensor data, which is then filtered through
 the original
 estimator. 
 An advantage of this approach is that users do not need to rely on a
 trusted server to maintain their privacy since they can themselves
 release noisy signals. The question is which of the two approaches
 can lead to a better estimation result for the same level of differential privacy.

 We can now compare these approaches numerically for the $W_2$-MHE
 estimator method. By selecting $s = 1$ and adding a Gaussian noise
 directly to both sets of adjacent sensor data, we re-run our test to
 find the trade-off between accuracy and privacy. The Gaussian noise
 has zero mean and the covariance matrix $Q = (1 -\bar{s}) (I +
 \frac{R + R'}{2})$,
 where, $I$ is the identity matrix and $R$ is a matrix of random
 numbers inl $(0,1)$. A value $\bar{s} \in [0,1]$ quantifies how much
 the sensor data is perturbed.  In simulation, we change the value of
 $\bar{s}$ and compare the results.

 We first set $\bar{s} = 0.944$ in $Q$, see
Figure~\ref{fig:NoiseInput56} (left figure). The $\epsilon_c$ and
 estimate error are found to be $\epsilon_c = 0.70306,
 E_{\text{correct}} = 0.0010771$ {\color{black}and $\lambda=0.1550$}. Then, we decrease $\bar{s}=0.894$, see
Figure~\ref{fig:NoiseInput56} (right figure). We obtain $\epsilon_c = 0.41844$, $E_{\text{correct}} = 0.0013998$ {\color{black}and $\lambda=0.1351$}. They also show that the higher level of differential privacy is achieved at the loss of accuracy. {\color{black}For the four cases, $\lambda$ values are not significant, indicating the approximation method is meaningful.}
 
\begin{figure}[t]
\centering
\begin{subfigure}{0.23\textwidth}
\centering
\includegraphics[height=3.6cm,width=4.2cm]{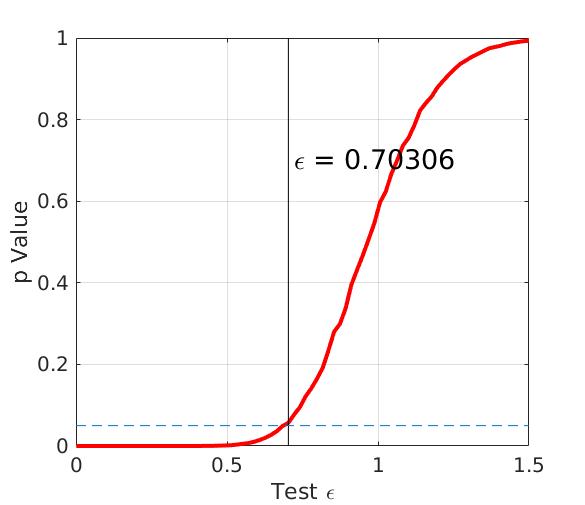}
\caption{$\bar{s} = 0.944$}
\end{subfigure}
\begin{subfigure}{0.23\textwidth}
\centering
\includegraphics[height=3.6cm,width=4.2cm]{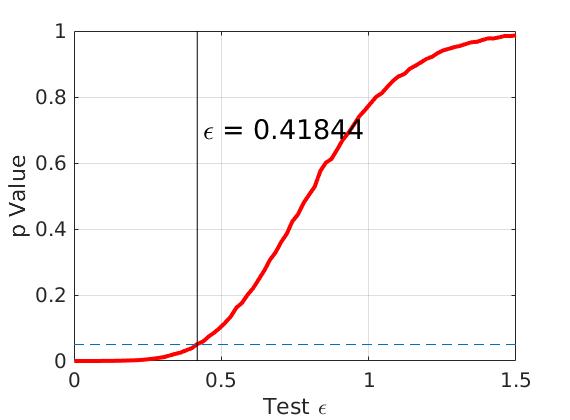}
\caption{$\bar{s} = 0.894$}
\end{subfigure}
\caption{Input perturbation: Privacy test results}
\label{fig:NoiseInput56}
\end{figure}

 From ($s=0.8$, $\epsilon_c=0.39947$, $E_{\text{correct}} =
 0.0040408$) and ($\bar{s}=0.894$, $\epsilon_c=0.41844$,
 $E_{\text{correct}} = 0.0013998$), we find that although the level of
 differential privacy is close to each other, the estimate error of
 $\bar{s}=0.894$ is only $1/3$ of that of $s=0.8$. Thus, the second
 mechanism (adding noise directly at the mechanism input) seems to
 indicate that {\color{black}it} can lead to better accuracy while maintaining the same
{\color{black}$(\epsilon$,$d$-adj$)$} differential privacy guarantee for this set of
 sensors.

 We further compare the performances of the two mechanisms on other
 sensor setups (e.g. uniformly located on the circle), see
 Table~\ref{table:compa}.



\vspace{-0.3cm}
\begin{table}[H]
\setlength{\abovecaptionskip}{-0cm}
\begin{center}
\begin{tabular}{ |c|c|c|c| } 
\hline
Sensor Setup & $W_2$MHE & Input Perturbation & Better choice \\
\hline
\multirow{3}{*}{$\mathbf{Q_2}$} & $\epsilon_c = 0.53229$ & $\epsilon_c = 0.72204$ & \multirow{3}{*}{$W_2$-MHE}  \\ 
& {\color{black}$\lambda = 0.1011$} & {\color{black}$\lambda$=0.2106}& \\
& $E_{\text{correct}} =
0.0049874$ & $E_{\text{correct}} =
0.0049674$ & \\
\hline
\multirow{3}{*}{$\mathbf{Q_3}$} & $\epsilon_c = 0.98768$ & $\epsilon_c = 2.3423$ & \multirow{3}{*}{$W_2$-MHE} \\
& {\color{black}$\lambda = 0.1037$} & {\color{black}$\lambda = 0.8408$}& \\
& $E_{\text{correct}} = 0.0030866$ & $E_{\text{correct}} =
0.0037826$  & \\
\hline
\end{tabular}
\end{center}
\caption{Comparisons of two mechanisms} \label{table:compa}
\end{table}
\vspace{-0.5cm}
From the table, we see that performance depend on the specific sensor
setups. In 2 out of 3 sensor setups, perturbing the filter output
seems the better option. However, more simulation results are needed
to reach a reliable conclusion. {\color{black}Besides, we see that the approximation method works only when differential privacy holds ($\epsilon_c$ is relatively small).}

\paragraph{Differentially private EKF {\color{black}\cite{JLN-GJP:14}}.}
The framework can also be applied to other differentially private
estimators. For comparison, we evaluate the performance of an extended
Kalman filter applied to the same examples. Compared
EKF, random noise $\frac{1-\hat{s}}{\hat{s}} w$ ($w$ is uniformly distributed over $[0,1]$) is added to the filter output at the update step, which makes the estimator differentially private. The initial guess $\mu_0$ is the same for the EKF and for the $W_2$-MHE.


Table~\ref{table:EKF} illustrates the EKF test results. 
Compared with those of $W_2$-MHE, the performance of the EKF remains
to be worse than that of $W_2$-MHE with respect to both privacy level
and RMSE (for all sensor setups). This is consistent with the fact
that $W_2$-MHE performs better than the EKF for multi-modal
distributions.  

\begin{table}[h]
\setlength{\abovecaptionskip}{-0cm}
\begin{center}
\begin{tabular}{ |c|c|c|c| } 
\hline
Sensor Setup & $\epsilon_c$ & $E_{\text{correct}}$ & Better choice\\
\hline
$\mathbf{Q_1} $ & $\quad 0.46223 \quad$ & $\quad 0.0066205 \quad$ & $W_2$-MHE \\
\hline
$\mathbf{Q_2} $ & $\quad 1.9239 \quad$ &$\quad 0.0064686 \quad$ & $W_2$-MHE \\
\hline
$\mathbf{Q_3} $ & $\quad 2.3085 \quad$ & $\quad 0.0062608 \quad$ & $W_2$-MHE \\
\hline
\end{tabular}
\end{center}
\caption{Diff-private EKF test results} \label{table:EKF}
\end{table}
\paragraph{Correctness of  sufficient condition for $W_2$-MHE.}
Theorem~4 of \cite{VK-SM:20} provides with a theoretical formula to
calculate $s$ that guarantees {\color{black}$(\epsilon$,$d$-adj$)$} differential
privacy.
Since this condition is derived using several assumptions and
upper bounds, the answer is in general expected to be conservative.

In order to make comparisons, we choose the sensor setup
$\mathbf{Q_2}$ and take $s=0.8$ for simulation. The value of other
parameters are: $T = 7 \text{ (for less computation)}, c_f = 1.0777,
c_h = 100, l = 1293.2 (N=5), \operatorname{diam}(\mathcal{K}_{0})=0.1,
{\color{black}d}=10$. Plug these values into the theorem, we can obtain
$\epsilon \ge 6474.3$. In simulation, the critical $\epsilon_c$ =
0.89281. Upon inspection, it is clear that the theoretical answer is
much more conservative than the approximated one, which indicates that
if $\epsilon \ge 0.89281$,  differential privacy is satisfied with
high confidence wrt the \textit{given space partition}. While this is
a necessary condition for privacy,  we run a  few more simulations to
test how this changes for finer space partitions.
\begin{enumerate}
\item $r = 3$ (9 regions per time step), $\epsilon_c = 0.98162$
\item $r = 4$ (16 regions per time step), $\epsilon_c = 1.9939$
\item $r \ge 5$, much more computation are required (numebr of events increases exponentially)
\end{enumerate}
As observed, a finer space partition leads an increase of
$\epsilon_c$.  But the theoretical bound is still far from the observed
values.

\vspace{-0.1cm}
\section{CONCLUSION} \label{section: discussion}
\vspace{-0.1cm}
 This work presents a
numerical test framework to evaluate the differential privacy of
continuous-range mechanisms such as state estimators. This includes a
precise quantification of its performance guarantees.  
Then, we apply the numerical method on differentially-private versions
of the $W_2$-MHE filter, and compare it with {\color{black}other} competing
approaches. Future work will be devoted to obtain more efficient
algorithms that; e.g.,~refine the considered partition
adaptively.

{\small
\bibliographystyle{./IEEEtran}
\bibliography{alias,SMD-add}
}
\end{document}